\documentclass[runningheads]{llncs}
\usepackage[utf8]{inputenc}
\usepackage[short]{optidef}
\usepackage{microtype}
\usepackage[vlined,ruled,algo2e,linesnumbered]{algorithm2e}
\usepackage{amsmath}
\usepackage{amsfonts}
\usepackage{tikz}
\usepackage{hyperref}

\newcommand{\bbE}{\mathbb{E}}
\newcommand{\bbP}{\mathbb{P}}
\newcommand{\bbN}{\mathbb{N}}

\newcommand{\bbR}{\mathbb{R}}

\newenvironment{proof*}[1]
  {\renewcommand{\proofname}{#1}%
   \begin{proof}}
  {\end{proof}}

\title{Online Matching with High Probability}
\author{Milena Mihail\thanks{The author was supported in part by a UCI-faculty startup grant.} \and Thorben Tr\"obst\thanks{The author was supported in part by NSF grant CCF-1815901. Part of this work was done while the author attended the ``Trimester Program on Discrete Optimization'' at the Hausdorff Institute for Mathematics in Bonn, Germany.}}
\institute{\texttt{\{mihail, t.troebst\}@uci.edu}\\ Department of Computer Science\\ University of California, Irvine, USA}

\begin{document}

\maketitle

\begin{abstract}
We study the classical, randomized \textsc{Ranking} algorithm which is known to be $(1 - \frac{1}{e})$-competitive in expectation for the Online Bipartite Matching Problem.
We give a tail inequality bound (Theorem~\ref{thm:obm_concentration}), namely that \textsc{Ranking} is $(1 - \frac{1}{e} - \alpha)$-competitive with probability at least $1 - e^{-2 \alpha^2 n}$ where $n$ is the size of the maximum matching in the instance.
Building on this, we show similar concentration results for the Fully Online Matching Problem and for the Online Vertex-Weighted Bipartite Matching Problem.

\end{abstract}

\section{Introduction}

In the Online Bipartite Matching Problem, we have an undirected, bipartite graph $G = (S, B, E)$ with a set $S$ of \emph{sellers} and a set $B$ of \emph{buyers}.
The buyers arrive online in adversarial order and every time a buyer $i$ arrives, all of its neighbors $N(i)$ are revealed.
We must then decide irrevocably which as of yet unmatched neighbor of $i$ should get matched to $i$.
The goal is to maximize the number of edges in the final matching $M$ relative to the size of the maximum matching in $G$, i.e.\ the so-called competitive ratio.

A matching $M$ is considered \emph{maximal} if there is no edge in $E$ which can be added to $M$ while preserving the matching property.
Due to the well-known fact that every maximal matching contains at least half of the edges of any \emph{maximum} matching, it is easy to see that any algorithm which matches arriving buyers whenever possible must be $\frac{1}{2}$-competitive.
Moreover, because of the adversarial arrival of buyers, this is best possible for deterministic algorithms.
In their seminal work, Karp et al.\ \cite{KarpVV/STOC/1990} defined the randomized \textsc{Ranking} algorithm (see Algorithm~\ref{alg:ranking}) and showed that it is $(1 - \frac{1}{e})$-competitive in expectation.
They also showed that this is the best possible competitive ratio for any randomized algorithm.

\medskip
\begin{algorithm2e}[H]
    Sample a uniformly random permutation $\pi$ on $S$. \\
    \For{each buyer $i$ who arrives}{
        Match $i$ to the first unmatched buyer in $N(i)$ wrt.\ to $\pi$.
    }
    \caption{\textsc{Ranking}\label{alg:ranking}}
\end{algorithm2e}
\medskip

In the years since then, online matching problems have received a large amount of interest due to the vast number of applications created by the internet and mobile computing.
Online advertising alone poses the AdWords Problem \cite{MehtaSVV/JACM/2007,HuangZZ/FOCS/2020,Vazirani/arXiv/2021} that lies at the heart of a multi-billion dollar market.
Another interesting application, the Fully Online Matching Problem \cite{HuangKTWZZ/STOC/2018,HuangPTTWZ/SODA/2019}, came about due to the rise of ride-sharing and/or ride-hailing apps such as Uber and Lyft.

For many online matching problems, there are extensions of \textsc{Ranking} which achieve competitive ratios of $1 - \frac{1}{e}$ or at the very least strictly greater than $\frac{1}{2}$.
Often, these are best-known for their respective problems.
However, to the best of our knowledge, all results on online matching in the literature only consider the competitive ratio \emph{in expectation} without guaranteeing any form of concentration beyond the trivial bounds that follow from Markov's inequality.

The analysis of concentration bounds for randomized algorithms goes back to the 1970s with classic results such as the second moment bound for \textsc{Quicksort} \cite{Sedgewick/Book/1975}.
See \cite{DubhashiP/Book/2009} for an extensive overview of the field.
However, it has remained a niche area of interest, including in the analysis of online algorithms where---with some exceptions (e.g. \cite{KommKKM/STACS/2014,LeonardiMPR/SIAMJoC/2001})---results are quantified in terms of expected solution quality only.

In some sense this is due to the well-known fact that, as a consequence of standard Chernoff bounds, any randomized algorithm which is good in expectation can be boosted to be good with probability $1 - \frac{1}{n}$ by simply repeating it $O(\log n)$ many times.
But it is precisely in the case of online algorithms where this argument fails due to the fact that online algorithms, by definition, can not be repeated.
Thus we argue that it is of particular interest to analyze the concentration of randomized online algorithms such as \textsc{Ranking}.

\subsection{Our Results and Techniques}

Our results concern \textsc{Ranking} type algorithms in three different settings: the classic Online Bipartite Matching Problem (see Section~\ref{sec:online_bipartite_matching}), the Fully Online Matching Problem inspired by ride-sharing (see Section~\ref{sec:fully_online_matching}) and the Online Vertex-Weighted Bipartite Matching Problem inspired by the internet advertising markets (see Section~\ref{sec:weighted}).

In Section~\ref{sec:online_bipartite_matching} we will show the following result, complementing the classic $\left(1 - \frac{1}{e}\right)$-competitiveness result of \textsc{Ranking} for the Online Bipartite Matching Problem \cite{KarpVV/STOC/1990} .

\begin{theorem}\label{thm:obm_concentration}
    Let $G = (S, B, E)$ be an instance of the Online Bipartite Matching Problem which admits a matching of size $n$.
    Then for any $\alpha > 0$ and any arrival order,
    \[
        \bbP\left[|M| < \left(1 - \frac{1}{e} - \alpha\right) n\right] < e^{- 2 \alpha^2 n}
    \]
    where $M$ is the random variable denoting the matching generated by \textsc{Ranking}.
\end{theorem}

The key technical ingredient for this result is a bounded differences property of the random variable $|M|$ (see Lemma~\ref{lem:bounded_diffs}).
We prove this via structural properties of matchings (see Lemma~\ref{lem:ranking_minus_j}) similar to ones which have been used in previous analyses of \textsc{Ranking} \cite{GoelM/SODA/2008,BirnbaumM/SIGACTNews/2008}.
Finally, this allows us to apply the following key lemma which is a consequence of Azuma's inequality.

\begin{lemma}[McDiarmid's Inequality \cite{McDiarmid/Book/1989}]\label{lem:mcdiarmid}
    Let $c_1, \ldots, c_n \in \bbR_+$ and consider some function $f : [0, 1]^n \rightarrow \bbR$ satisfying
    \[
        |f(x_1, \ldots, x_{i - 1}, x'_i, x_{i + 1}, \ldots, x_n) - f(x_1, \ldots, x_n)| \leq c_i
    \]
    for all $x \in [0, 1]^n$, $i \in [n]$ and $x'_i \in [0, 1]$.
    Moreover let $\Delta^n$ be the uniform distribution on $[0, 1]^n$.
    Then for all $t > 0$, we have
    \[
        \bbP_{x \sim \Delta^n}[f(x) < \bbE_{y \sim \Delta^n}[f(y)] - t] < e^{- \frac{2 t^2}{\sum_{i = 1}^n{c_i^2}}}.
    \]
\end{lemma}

In Section~\ref{sec:fully_online_matching} we will define the Fully Online Matching Problem and the natural extension of \textsc{Ranking} for this setting.
We will then give a similar concentration bound as in Theorem~\ref{thm:obm_concentration}.

\begin{theorem}\label{thm:fom_concentration}
    Let $G$ be an instance of the Fully Online Matching Problem which admits a matching of size $n$.
    Then for any $\alpha > 0$,
    \[
        \bbP\left[|M| < \left(\rho - \alpha\right) n\right] < e^{- \alpha^2 n}
    \]
    where $M$ is the random variable denoting the matching generated by \textsc{Ranking} and $\rho$ is the competitive ratio of \textsc{Ranking}. 
\end{theorem}

We remark that by \cite{HuangKTWZZ/STOC/2018}, we know $\rho > 0.521$ and for the special case where $G$ is bipartite, we have $\rho = W(1) \approx 0.567$.

Lastly, in Section~\ref{sec:weighted} we will consider the Online Vertex-Weighted Bipartite Matching Problem.
In this setting, a generalization of \textsc{Ranking} was shown to be $(1 - \frac{1}{e})$-competitive by Aggarwal et al.\ \cite{AggarwalGKM/SODA/2011}.
We will modify this algorithm to show the following.

\begin{theorem}\label{thm:ovwbm_concentration}
    For any $\alpha > 0$, there exists a variant of \textsc{Ranking} such that for any instance $G = (S, B, E)$ with weights $w : S \rightarrow \bbR_+$ of the Online Vertex-Weighted Bipartite Matching, any arrival order of $B$ and any matching $M^*$,
    \[
        \bbP\left[w(M) < \left(1 - \frac{1}{e} - \alpha\right) w(M^*)\right] < e^{- \frac{1}{50} \alpha^4 \frac{w(M^*)^2}{||w||_2^2}}
    \]
    where $M$ denotes the matching generated by \textsc{Ranking} and
    \[
        w(M) \coloneqq \sum_{\{i, j\} \in M} w_j.
    \]
\end{theorem}

\section{Online Bipartite Matching}\label{sec:online_bipartite_matching}

In order to analyze \textsc{Ranking}, it is common to replace the sampling of the permutation $\pi$ in Algorithm~\ref{alg:ranking} by sampling an independent, uniform $x_j \in [0, 1]$ for every $j \in S$ called the \emph{rank} of $j$.
Then, sorting $S$ by the values of $x_j$ yields a uniformly random permutation.
Formally, this is Algorithm~\ref{alg:ranking_2}.

\medskip
\begin{algorithm2e}[H]
    \For{$j \in S$}{
        Sample a uniformly random $x_j \in [0, 1]$.
    }
    \For{each buyer $i$ who arrives}{
        Match $i$ to an unmatched $j \in N(i)$ minimizing $x_j$.
    }
    \caption{\textsc{Ranking}\label{alg:ranking_2}}
\end{algorithm2e}
\medskip

In the following, consider a fixed graph $G = (S, B, E)$ with a fixed arrival order.
Assume that $|S| = n$ and that $G$ has a matching of size $n$.
We define a function $f : [0, 1]^S \rightarrow \bbR$ by letting $f(y)$ be the size of the matching $M$ generated by Algorithm~\ref{alg:ranking_2} if $x_j = y_j$ for all $j \in S$.
Our goal will then be to show the following Lemma which is a different perspective on a structural property that appears under various forms in the online matching literature (e.g.\ Lemma~2 in \cite{BirnbaumM/SIGACTNews/2008}).

\begin{lemma}[Bounded Differences]\label{lem:bounded_diffs}
    Let $x \in [0, 1]^S$, $j^\star \in S$ and $\theta \in [0, 1]$ be arbitrary.
    Define $x'_j$ to be $\theta$ if $j = j^\star$ and $x_j$ otherwise.
    Then $|f(x) - f(x')| \leq 1$.
\end{lemma}

Note that Lemma~\ref{lem:bounded_diffs} implies Theorem~\ref{thm:obm_concentration} via McDiarmid's inequality (Lemma~\ref{lem:mcdiarmid}).
Specifically, by applying McDiarmid to the function $f$ with $c \equiv 1$ we get
\begin{align*}
    \bbP\left[|M| < \left(1 - \frac{1}{e} - \alpha\right) n\right] &\leq \bbP_{x \sim \Delta^S}[f(x) < \bbE_{y \sim \Delta^S}[f(y)] - \alpha n] \\
        &\leq e^{-\alpha n^2}
\end{align*}
where we used that $(1 - \frac{1}{e}) n \leq \bbE_{y \sim \Delta^S}[f(y)]$ since \textsc{Ranking} is $(1 - \frac{1}{e})$-competitive.
It remains to prove Lemma~\ref{lem:bounded_diffs}.

\begin{lemma}\label{lem:ranking_minus_j}
    Let $j \in S$, then we can define the graph $G_{-j}$ which contains all vertices of $G$ except for $j$.
    For some fixed values of $x \in [0, 1]^S$, we let $M$ be the matching produced by \textsc{Ranking} in $G$ and let $M_{-j}$ be the matching produced by \textsc{Ranking} in $G_{-j}$.
    Then $|M_{-j}| \leq |M| \leq |M_{-j}| + 1$.
\end{lemma}

\begin{proof}
    For any buyers $i, i' \in B$, let $N^{(i)}(i')$ be the set of neighbors of $i'$ in $G$ which are unmatched by the time that $i$ arrives in the run of \textsc{Ranking} with the fixed values of $x$.
    Likewise, let $N^{(i)}_{-j}(i')$ be the set of unmatched neighbors of $i'$ in the run of \textsc{Ranking} on $G_{-j}$ when $i$ arrives.
    We claim that for all $i \in B$ there exists some $j' \in S$ such that for all $i' \in B$ we have $N^{(i)}(i') = N^{(i)}_{-j}(i')$ or $N^{(i)}(i') = N^{(i)}_{-j}(i') \cup \{j'\}$.
    
    Let us show this claim via induction on $i \in B$ in order of arrival.
    Note that when the first buyer arrives, this holds for $j' = j$ because we have removed only $j'$ from the graph and nobody has been matched yet.
    Now assume that the statement holds when $i$ arrives, we need to see that it still holds after $i$ has been matched.
    Clearly, if $i$ gets matched to the same vertex in $G$ and in $G_{-j}$, then the inductive step follows trivially.
    
    So now assume that $i$ gets matched to different vertices in $G$ and in $G_{-j}$.
    By the inductive hypothesis this can only happen if $i$ gets matched to $j'$ in $G$ and it gets matched to some other $j''$ (potentially $j'' = \bot$, i.e.\ it is not matched at all) in $G_{-j}$.
    But then $N^{(i + 1)}(i') = N_{-j}^{(i + 1)}(i')$ or $N^{(i + 1)}(i') = N_{-j}^{(i + 1)}(i') \cup \{j''\}$ for all $i' \in B$.
    Thus the claim holds by induction.
    
    Finally, let us see that the claim implies the lemma.
    First note that since $i$ always has more unmatched neighbors in $G$ than in $G_{-j}$, we have $|M| \geq |M_{-j}|$.
    But on the other hand, if at some time in the algorithm $i$ is matched to $j'$ in $G$ and not matched at all in $G_{-j}$, then we have that $N^{(i + 1)}(i') = N_{-j}^{(i + 1)}(i')$ for all $i' \in B$.
    Thus the two runs will be identical from that point onward and $|M| = |M_{-j}| + 1$.\qed
\end{proof}

Finally, we can show that this implies the bounded differences property of $f$ that we claimed in Lemma~\ref{lem:bounded_diffs}.

\begin{proof*}{Proof of Lemma~\ref{lem:bounded_diffs}}
    By Lemma~\ref{lem:ranking_minus_j} we know that removing a seller from the graph can decrease the size of the matching computed by \textsc{Ranking} by at most one assuming that the values of the $x_j$ are fixed.
    But of course if we are removing $j^\star \in S$, the matching $M_{-j^\star}$ computed by \textsc{Ranking} in $G_{-j^\star}$ does not depend on the value of $x_{j^\star}$ or $x'_{j^\star}$.
    So we have
    \[
        |M_{-j^\star}| \leq f(x) \leq |M_{-j^\star}| + 1
    \]
    and
    \[
        |M_{-j^\star}| \leq f(x') \leq |M_{-j^\star}| + 1
    \]
    which implies $|f(x) - f(x')| \leq 1$ as claimed.\qed
\end{proof*}

As we have already seen, this is enough to prove Theorem~\ref{thm:fom_concentration} in the case where $|S| = n$.
To prove the general case we can use a simple reduction.
In particular, assuming that there is a matching $M$ of size $n$ but $|S| > n$, let $S_M$ be the sellers covered by $M$ and let $G_M = (S_M, B, E)$.
We have seen in Lemma~\ref{lem:ranking_minus_j} that for any fixed $x \in [0, 1]^S$, \textsc{Ranking} will produce a matching in $G$ that is not smaller than the matching it produces in $G_M$ when run with $x$ restricted to $S_M$.
Therefore, Theorem~\ref{thm:fom_concentration} on $G_M$ implies Theorem~\ref{thm:fom_concentration} on $G$ which establishes the general case.

\section{Fully Online Matching}\label{sec:fully_online_matching}

In the Fully Online Matching Problem we have a not necessarily bipartite graph $G$ the vertices of which arrive and depart online in adversarial order.
When a vertex arrives, it reveals all of its edges to vertices that have already arrived.
By the time it departs we are guaranteed to have been revealed its entire neighborhood.

This problem was introduced by Huang et al.\ \cite{HuangKTWZZ/STOC/2018} and is motivated by ride-sharing.
Each vertex represents a rider who, upon arrival, is willing to wait only for a certain amount of time.
Two riders can only be matched if the time that they spend on the platform overlaps, even in the offline solution.
This additional condition allows Huang et al.\ to show that the generalization of \textsc{Ranking} shown in Algorithm~\ref{alg:ranking_3} is $0.521$-competitive in general and $0.567$-competitive on bipartite graphs.

\medskip
\begin{algorithm2e}[H]
    \For{vertex $i$ who arrives}{
        Sample a uniformly random $x_i \in [0, 1]$.
    }
    \For{vertex $i$ who departs}{
        Match $i$ to an unmatched $j \in N(i)$ minimizing $x_j$.
    }
    \caption{\textsc{Fully Online Ranking}\label{alg:ranking_3}}
\end{algorithm2e}
\medskip

In order to show a concentration bound, we can apply similar techniques as in Section~\ref{sec:online_bipartite_matching}.
Let $G = (V, E)$ be a graph which admits a perfect matching of size $n$.
Then let $f : [0, 1]^V \rightarrow \bbR$ represent once again the size of the matching generated by Algorithm~\ref{alg:ranking_3} when given the $x_i$ values.
The corresponding bounded differences condition then becomes:

\begin{lemma}[Bounded Differences]\label{lem:bounded_diffs_2}
    Let $x \in [0, 1]^V$, $i^\star \in V$ and $\theta \in [0, 1]$ be arbitrary.
    Define $x'_i$ to be $\theta$ if $i = i^\star$ and $x_i$ otherwise.
    Then $|f(x) - f(x')| \leq 1$.
\end{lemma}

This implies Theorem~\ref{thm:fom_concentration} as before though note that this time we will lose a factor of 2 since we now have $2 n$ variables.
We remark that this follows directly from Lemma~2.3 in \cite{HuangKTWZZ/STOC/2018} but for completeness we will give a short proof sketch.

\begin{lemma}\label{lem:ranking_minus_j_2}
    Using the notation from Lemma~\ref{lem:ranking_minus_j}, we have $|M_{-j}| \leq |M| \leq |M_{-j}| + 1$ for any $j \in V$ and fixed values of $x \in [0, 1]^V$.
\end{lemma}

\begin{proof}
    As in the proof of Lemma~\ref{lem:ranking_minus_j}, let $N^{(i)}(i')$ (or $N^{(i)}_{-j}(i')$) be the set of neighbors of $i'$ in $G$ (or $G_{-j}$) which is unmatched by the time that $i$ departs in the run of \textsc{Fully Online Ranking} with the fixed values of $x$.
    We claim that for all $i \in V$, there exists some $j' \in V$ such that for all $i' \in V$, we have $N^{(i)}(i') = N^{(i)}_{-j}(i')$ or $N^{(i)}(i') = N^{(i)}_{-j}(i') \cup \{j'\}$.
    
    This claim follows via an almost identical induction as in Lemma~\ref{lem:ranking_minus_j}.
    Then, since $i$ always has more unmatched neighbors in $G$ than in $G_{-j}$, we have $|M| \geq |M_{-j}|$.
    And if at some time in the algorithm, $i$ is matched to $j'$ in $G$ and not matched at all in $G_{-j}$, then we have that $N^{(i + 1)}(i') = N_{-j}^{(i + 1)}(i')$ for all $i' \in V$.
    Thus the two runs will the identical from that point onward and $|M| = |M_{-j}| + 1$.\qed
\end{proof}

Since Lemma~\ref{lem:ranking_minus_j_2} implies Lemma~\ref{lem:bounded_diffs_2}, this yields Theorem~\ref{thm:fom_concentration} for graphs which contain a perfect matching.
But as in Section~\ref{sec:online_bipartite_matching}, we may drop this condition by reducing a graph $G$ with a matching $M$ to the subgraph induced by the vertices covered by $M$.
Adding the vertices back in only increases the performance of \textsc{Fully Online Ranking} by Lemma~\ref{lem:ranking_minus_j_2}.

\section{Online Vertex-Weighted Bipartite Matching}\label{sec:weighted}

In this section we will consider a weighted extension of the Online Bipartite Matching Problem which has been inspired by online advertising markets.
In the Online Vertex-Weighted Bipartite Matching Problem, we have a bipartite graph $G = (S, B, E)$ with vertex weights $w : S \rightarrow \bbR_+$ on the offline vertices.
Here $S$ represents the advertisers and $B$ represents website impressions or search queries which should get matched to ads from the advertisers.
The vertices $B$ arrive online in adversarial order and should get matched to a neighbor $j$ such that the total weight of the matched vertices in $S$ is maximized.
This problem can be seen as a special case of the AdWords Problem which instead imposes edge-weights and budgets on the offline vertices.

Perhaps somewhat surprisingly it took 20 years for \textsc{Ranking} to be extended for the unweighted to the vertex-weighted setting by Aggarwal et al.\ \cite{AggarwalGKM/SODA/2011}.
This is because in the presence of weights, it is no longer enough to pick a uniformly random permutation over the offline vertices.
Instead, one has to skew the permutation so that heavier vertices are more likely to appear first.
This is done elegantly in Algorithm~\ref{alg:ranking_4} by ordering the vertices not by their $x_j$ but rather by the careful chosen quantity $w_j (1 - e^{x_j - 1})$.

\medskip
\begin{algorithm2e}[H]
    \For{$j \in S$}{
        Sample a uniformly random $x_j \in [0, 1]$.
    }
    \For{each buyer $i$ who arrives}{
        Match $i$ to an unmatched $j \in N(i)$ maximizing $w_{j} \left(1 - e^{x_{j} - 1}\right)$.\label{alg:ranking_4:max}
    }
    \caption{\textsc{Vertex-Weighted Ranking}\label{alg:ranking_4}}
\end{algorithm2e}
\medskip

However, Algorithm~\ref{alg:ranking_4} does not lend itself to a straight-forward analysis via the method of bounded differences.
This is because a vertex with small weight, which should have little impact on the total weight of the matching, can sometimes be chosen over a vertex with much larger weight.
See the example shown in Figure~\ref{fig:bad_concentration}.

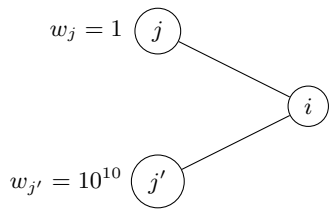
\begin{figure}[hbt]
    \centering
    \begin{tikzpicture}
        \node[circle, draw] (j) at (0, 0) {$j$};
        \node[circle, draw] (jp) at (0, -2) {$j'$};
        \node[circle, draw] (i) at (2, -1) {$i$};
        
        \draw[-] (j) -- (i);
        \draw[-] (jp) -- (i);
        
        \node[left, anchor=east, xshift=-10pt] at (j) {$w_j = 1$};
        \node[left, anchor=east, xshift=-10pt] at (jp) {$w_{j'} = 10^{10}$};
    \end{tikzpicture}
    \caption{Shown is a simple instance in which the value of $x_j$ can have a large impact on the final matching despite the fact that $w_j$ is small. If $x_{j'} \gg 1 - 10^{-10}$, $i$ will choose $j$ in line~\ref{alg:ranking_4:max} for sufficiently small values of $x_j$.}
    \label{fig:bad_concentration}
\end{figure}

In particular, the problem lies with the fact that $w_j (1 - e^{x_j - 1})$ can get arbitrarily close to 0 if $x_j$ gets close to 1.
We will overcome this problem by changing the function slightly.
For any $\epsilon > 0$ we consider $\epsilon$-\textsc{Ranking} as shown in Algorithm~\ref{alg:ranking_5}.

\medskip
\begin{algorithm2e}[H]
    \For{$j \in S$}{
        Sample a uniformly random $x_j \in [0, 1]$.
    }
    \For{each buyer $i$ who arrives}{
        Match $i$ to an unmatched $j \in N(i)$ maximizing $w_{j} \left(1 - e^{x_{j} - 1 - \epsilon}\right)$.\label{alg:ranking_5:max}
    }
    \caption{$\epsilon$-\textsc{Ranking}\label{alg:ranking_5}}
\end{algorithm2e}
\medskip

In the following fix some instance $G = (S, B, E)$ with vertex-weights $w$ and some $\epsilon > 0$.
Then we let $f : [0, 1]^S \rightarrow \bbR$ represent the total weight of the matching generated by Algorithm~\ref{alg:ranking_5} with fixed samples $x_j$.
We will show that $\epsilon$-\textsc{Ranking} is still $(1 - \frac{1}{e} - \epsilon)$-competitive while also allowing us to give a concentration bound.

To give a concise proof of the $(1 - \frac{1}{e} - \epsilon)$-competitiveness we will use the economic analysis by Eden et al.\ \cite{Eden/SOSA/2020} which is itself based on the primal-dual viewpoint due to Devanur et al.\ \cite{DevanurJK/SODA/2013}.
This analysis associates random variables $r_j$ with all $j \in S$ and $u_i$ with $i \in B$.
The idea is that the value $w_j e^{x_j - 1 - \epsilon}$ represents the \emph{price} of $j$ and whenever a match between $i$ and $j$ is made, this is a \emph{sale}.
We will then set $r_j$ (the \emph{revenue}) to be $w_j e^{x_j - 1 - \epsilon}$ and $u_i$ (the \emph{utility}) to be $w_j e^{x_j - 1 - \epsilon}$.
If a vertex is never matched, its revenue / utility will be zero.

\begin{lemma}\label{lem:ranking_minus_j_3}
    Using the notation from Lemma~\ref{lem:ranking_minus_j}, we have that for all $j \in S$ and fixed samples $x$,
    \[
        w(M_{-j}) - \frac{2}{\epsilon} w_j \leq w(M) \leq w(M_{-j}) + w_j.
    \]
    Additionally, for any $i \in B$, its utility $u_i$ in the run on $G$ will be no less than in the run on $G_{-j}$.
\end{lemma}

\begin{proof}
    For any buyers $i, i' \in B$, let $N^{(i)}(i')$ be the set of neighbors of $i'$ in $G$ which are unmatched by the time that $i$ arrives in the run of Algorithm~\ref{alg:ranking_5} with the fixed values of $x$.
    Likewise, let $N^{(i)}_{-j}(i')$ be the set of unmatched neighbors of $i'$ in the run of $\epsilon$-\textsc{Ranking} on $G_{-j}$ when $i$ arrives.
    We claim that for all $i \in B$ there exists some $j' \in S$ such that
    \[
        w_{j'} (1 - e^{x_{j'} - 1 - \epsilon}) \leq w_j (1 - e^{x_j - 1 - \epsilon})
    \]
    and for all $i' \in B$, we have $N^{(i)}(i') = N^{(i)}_{-j}(i')$ or $N^{(i)}(i') = N^{(i)}_{-j}(i') \cup \{j'\}$.
    
    This claim is almost the same as in the proof of Lemma~\ref{lem:ranking_minus_j} and may likewise be shown via induction.
    Note that the extra condition on $w_{j'}$ holds at the beginning where $j' = j$ and continues to hold throughout the induction because whenever $i$ matches to $j'$, it can only free up 
    This clearly holds at the beginning where $j' = j$ and whenever $i$ matches to $j'$, it frees up a vertex $j''$ with
    \[
        w_{j''} (1 - e^{x_{j''} - 1 - \epsilon}) \leq w_{j'} (1 - e^{x_{j'} - 1 - \epsilon})
    \]
    due to the fact that $j'$ was picked over $j''$ in line~\ref{alg:ranking_5:max}.
    If $i$ was not even matched in $G_{-j}$, we can simply set $j' = j$ for the induction.
    
    Now note that since $N^{(i)}_{-j}(i) \subseteq N^{(i)}(i)$ for all $i \in B$, we always maximize over a larger set in line~\ref{alg:ranking_5:max}.
    Thus the utility of $i$ will be no smaller in the run on $G$ compared to the run on $G_{-j}$.
    
    On the other hand, let $T \subseteq S$ be the set of sellers matched in the run on $G$ and let $T_{-j} \subseteq S \setminus \{j\}$ be the set of sellers matched in the run on $G_{-j}$.
    Then we observe that $T \setminus T_{-j} \subseteq \{j\}$ because for all $j' \neq j$, if $j'$ gets matched to $i$ in $M$, then either $j' \in N^{(i)}_{-j}(i)$ implying that $i$ will match to $j'$ in $M_{-j}$ , or $j'$ was already matched to some other vertex.
    In both cases, if $j' \in T$ then $j' \in T_{-j}$.
    This implies that $w(M) \leq w(M_{-j}) + w_j$.
    
    We also have that $|T_{-j} \setminus T| \leq 1$.
    Simply imagine a buyer $i^\star$ that arrives after all other buyers and has edges to all sellers.
    Then by the claim, there exists some $j' \in S$ such that
    \[
        (S \setminus \{j\}) \setminus T_{-j} = N^{(i^\star)}_{-j}(i^\star) \subseteq N^{(i^\star)}(i^\star) \cup \{j'\} = S \setminus (T \cup \{j'\})
    \]
    and so $T_{-j} \subseteq T \cup \{j'\}$.
    This implies that $w(M) \geq w(M_{-j}) - w_{j'}$.
    
    Finally, we also know by the claim that $w_{j'} (1 - e^{x_{j'} - 1 - \epsilon}) \leq w_j (1 - e^{x_j - 1 - \epsilon})$ which implies
    \[
        w_{j'} \leq \frac{1}{1 - e^{-\epsilon}} w_j \leq \frac{1}{\left(1 - \frac{1}{e}\right) \epsilon} w_j \leq \frac{2}{\epsilon} w_j.
    \]
    Thus we have shown $w(M_{-j}) - \frac{2}{\epsilon} w_j \leq w(M) \leq w(M_{-j}) + w_j$ as required.\qed
\end{proof}

\begin{lemma}[Bounded Differences]\label{lem:bounded_diffs_3}
    Let $x \in [0, 1]^S$, $j^\star \in S$ and $\theta \in [0, 1]$ be arbitrary.
    Define $x'_j$ to be $\theta$ if $j = j^\star$ and $x_j$ otherwise.
    Then $|f(x) - f(x')| \leq \left(1 + \frac{2}{\epsilon}\right) w_{j^\star}$.
\end{lemma}

\begin{proof}
    As in the proof of Lemma~\ref{lem:bounded_diffs}, we can simply remove $j^\star$ and apply Lemma~\ref{lem:ranking_minus_j_3}.
    Then
    \begin{align*}
        w(M_{-j^\star}) - \frac{2}{\epsilon} w_{j^\star} &\leq f(x) \leq w(M_{-j^\star}) + w_{j^\star}, \\
        w(M_{-j^\star}) - \frac{2}{\epsilon} w_{j^\star} &\leq f(x') \leq w(M_{-j^\star}) + w_{j^\star}
    \end{align*}
    which implies the result.\qed
\end{proof}

\begin{lemma}\label{lem:dual_feasible}
    For any $\{i, j\} \in E$, we have $\bbE[r_j + u_i] \geq (1 - \frac{1}{e} - \epsilon) w_j$.
\end{lemma}

\begin{proof}
    Fix all samples $x$ except for $x_j$.
    Then we can define $u^*$ to be the utility of $i$ when $\epsilon$-\textsc{Ranking} is ran on $G_{-j}$.
    By Lemma~\ref{lem:ranking_minus_j_3}, we know that $u_i \geq u^*$, regardless of the value of $x_j$.
    
    On the other hand, if $x_j$ is small enough that $w_j (1 - e^{x_j - 1 - \epsilon}) > u^*$, then $j$ will definitely get matched because if $j$ is not yet matched by the time that $i$ arrives, then clearly $j$ will be chosen in line~\ref{alg:ranking_5:max} of the algorithm and so it gets matched to $i$.
    Now if $u^*$ is very small, this may be the case for all values of $x_j$ and in that case
    \[
        \bbE[r_j \mid x_{-j}] \geq \int_0^{1} w_j e^{t - 1 - \epsilon}\, \mathrm{d} t = \left(1 - \frac{1}{e}\right) e^{-\epsilon} w_j \geq \left(1 - \frac{1}{e} - \epsilon\right) w_j.
    \]
    Otherwise there will be some value $z \in [0, 1]$ such that $w_j (1 - e^{z - 1 - \epsilon}) = u^*$ and then we can compute
    \[
        \bbE[r_j \mid x_{-j}] \geq \int_0^{z} w_j e^{t - 1 - \epsilon}\, \mathrm{d} t = \left(1 - \frac{1}{e}\right) w_j - u^*.
    \]
    But clearly, in both cases we have
    \[
        \bbE[r_j + u_i \mid x_{-j}] \geq \bbE[r_j \mid x_{-j}] + u^* \geq \left(1 - \frac{1}{e} - \epsilon\right) w_j
    \]
    and so in particular $\bbE[r_j + u_i] \geq (1 - \frac{1}{e} - \epsilon) w_j$ as claimed.\qed
\end{proof}

\begin{lemma}\label{lem:ranking_competitive}
    $\epsilon$-\textsc{Ranking} is $(1 - \frac{1}{e} - \epsilon)$-competitive.
\end{lemma}

\begin{proof}
    Let $M^*$ be a maximum weight matching and let $M$ be the matching output by $\epsilon$-\textsc{Ranking}.
    Notice that every time we match an edge in the algorithm, we increase $\sum_{j \in S} r_j + \sum_{i \in B} u_i$ by exactly the weight of the edge.
    Thus by Lemma~\ref{lem:dual_feasible},
    \begin{align*}
        \bbE[w(M)] &= \bbE\left[\sum_{j \in S} r_j + \sum_{i \in B} u_i\right] \geq \sum_{\{i, j\} \in M^*} \bbE[r_j + u_i] \\
            &\geq \sum_{\{i, j\} \in M^*} \left(1 - \frac{1}{e} - \epsilon\right) w_j = \left(1 - \frac{1}{e} - \epsilon\right) w(M^*)
    \end{align*}
    and therefore $\epsilon$-\textsc{Ranking} is $(1 - \frac{1}{e} - \epsilon)$-competitive.\qed
\end{proof}

Finally, we have the tools necessary to show Theorem~\ref{thm:ovwbm_concentration} by combining Lemma~\ref{lem:bounded_diffs_3} with Lemma~\ref{lem:ranking_competitive}.

\begin{proof*}{Proof of Theorem~\ref{thm:ovwbm_concentration}}
    Given some $\alpha > 0$, we consider the algorithm $\frac{\alpha}{2}$-\textsc{Ranking} which we know to be $(1 - \frac{1}{e} - \frac{\alpha}{2})$-competitive by Lemma~\ref{lem:ranking_competitive}.
    We apply Lemma~\ref{lem:mcdiarmid} (McDiarmid's inequality) with Lemma~\ref{lem:bounded_diffs_3} (bounded differences).
    This gives us
    \begin{align*}
        \bbP\left[w(M) < \left(1 - \frac{1}{e} - \alpha\right) w(M^*)\right] &< e^{-2 \frac{\alpha^2}{2} \frac{w(M^*)^2}{(1 + 4 / \alpha)^2 ||w||_2^2}} \\
            &\leq e^{-\frac{\alpha^4}{50} \frac{w(M^*)^2}{||w||_2^2}}
    \end{align*}
    where we use that $\alpha < 1$ since otherwise the bound holds trivially. \qed
\end{proof*}

The results of this section may also be extended to a generalization of the Online Vertex-Weighted Bipartite Matching Problem which is called the Online Single-Valued Bipartite Matching Problem.
The setup is almost identical in that we still have a bipartite graph $G = (S, B, E)$ with vertex weights $w : S \rightarrow \bbR_+$ on the offline vertices.
However, now each offline vertex $j$ also has a capacity $c_j \in \bbN$ that represents how often it is allowed to be matched.

Clearly, Theorem~\ref{thm:ovwbm_concentration} can be extended to this setting by simply creating $c_j$ many copies of each offline vertex $j$.
This can be done implicitly and in a capacity-oblivious way by simply sampling a new $x_j$ every time $j$ is matched during the \textsc{Ranking} (or $\epsilon$-\textsc{Ranking}) algorithm.

Recently, Vazirani \cite{Vazirani/arXiv/2021} showed that this ``resampling'' is in fact not necessary, i.e.\ that the same value of $x_j$ can be used for every copy of $j$ while still achieving $(1 - \frac{1}{e})$-competitiveness of \textsc{Ranking}; see Algorithm~\ref{alg:ranking_6}.

\medskip
\begin{algorithm2e}[H]
    \For{$j \in S$}{
        Sample a uniformly random $x_j \in [0, 1]$.
    }
    \For{each buyer $i$ who arrives}{
        Match $i$ to a $j \in N(i)$ which has been matched less than $c_j$ times, maximizing $w_{j} \left(1 - e^{x_{j} - 1}\right)$.\label{alg:ranking_6:max}
    }
    \caption{\textsc{Single-Valued Ranking}\label{alg:ranking_6}}
\end{algorithm2e}
\medskip

The main benefit of Algorithm~\ref{alg:ranking_6} is that it uses fewer random bits than running \textsc{Ranking} on the reduced instance with $c_j$ many copies of each offline vertex $j$.
However, it will accordingly be less tightly concentrated which leads to a version of Theorem~\ref{thm:ovwbm_concentration} in which the bound depends not on $||w||_2^2$ but rather on $\sum_{j} (c_j w_j)^2$.
Thus we will not pursue this approach here.

\section*{Acknowledgements}

We would like to thank Vijay Vazirani for helpful comments and feedback.

\bibliographystyle{splncs04}
\bibliography{references}

\end{document}